\documentclass[10pt,english]{article}

\usepackage{preamble}
\usepackage{macros}

\begin{document}
\maketitle

\begin{abstract}
  We present an analytical solution to the angle-finding problem in quantum signal processing (QSP) for monomials of odd degree.
  Specifically, we show that to implement a monomial of degree \( n \), where \( n \) is odd, it suffices to choose powers of a primitive \( n \)-th root of unity as QSP phase angles.
  Our approach departs from standard numerical methods and is rooted in a group-theoretic argument.
  Being fully analytical, it eliminates numerical errors and reduces computational overhead in QSP implementation of odd monomials.
  Such use cases arise, for example, in quantum computing, where self-adjoint contractions are embedded into unitary operators acting on extended Hilbert spaces.
\end{abstract}

\section{Introduction}

Within the framework of quantum signal processing (QSP) we show that monomials of odd degree \( n \) can be implemented by choosing the QSP phase angles as powers of a primitive \( n \)-th root of unity.
More precisely, for functional variables \( k \in \bbC^{2}\) and \( z \in \bbC \) define the matrices,
\begin{align}
  T(k) =
  \begin{pmatrix}
    k_{1} & k_{2} \\
    k_{2} & k_{1}
  \end{pmatrix},
   &  & \text{and} &  &
  S(z) =
  \begin{pmatrix}
    z & 0       \\
    0 & \bar{z}
  \end{pmatrix},
\end{align}
where \( \bar{z} \) denotes the complex conjugate of \( z \).
We prove the following result.
\begin{theorem}
  \label{thm: main}
  For \( n \in \bbN \) odd and \( \omega \) a primitive \( n \)-th root of unity we have,
  \begin{align}
    \label{eq: desired product}
    \prod_{i=1}^{n}T(k)S(\omega^{i})
    =
    T(k)S(\omega)T(k)S(\omega^{2})T(k)\cdots S(\omega^{n-1})T(k)
    =
    \begin{pmatrix}
      k_{1}^{n} & \ast \\
      \ast      & \ast
    \end{pmatrix}
    ,
  \end{align}
  where \( \ast \) stands for some (possibly different) polynomials in \( k_{1},k_{2} \).
  In particular, it is always possible to choose \( \omega = e^{i \frac{2\pi}{n}} \).
\end{theorem}
To our knowledge, this is the first closed-form exact expression for QSP phase angles implementing odd monomials.

In the original formulation of QSP \cite{PhysRevX.6.041067,low2017optimal} and its generalization \cite{gilyen2019quantum}, the QSP framework was introduced to implement arbitrary polynomials via suitable choices of QSP angles --- phases used in the \( S \)-matrices in the sequence of the form \eqref{eq: desired product}.
These angles are typically computed via numerical algorithms based on iterative constructions of complementary polynomials.
This approach has been further developed in several works (e.g., \cite{chao2020finding, haah2019product, skelton2024mostly, alase2025quantum}), and standard algorithms are now implemented in software libraries such as \texttt{pyqsp} \cite{pyqsp}.
For the recent summary of approaches see e.g. \cite{skelton2025hitchhiker}.
Despite these advances, in the case of odd monomials all known methods for computing QSP angles are numerical, introducing approximation errors and requiring computational resources.

Our result provides an exact, analytical construction that bypasses numerical synthesis entirely. This removes approximation errors and runtime overhead.
Since monomial implementations are a common building block in ground-state preparation using probabilistic imaginary time evolution \cite{lin:2021, kosugi2021probabilistic, Liu:2021aa,PhysRevA.109.052414}, our result offers a practical and efficient alternative for these applications.

\section{Methodology}

The proof of Theorem \ref{thm: main} is conceptually simple but technical.
To avoid obscuring the main idea, we outline the core reasoning below and provide the full proof in the next two sections.

\paragraph{Idea of proof.}
We analyze the upper-left entry of the product \eqref{eq: desired product}.
A priori this is a polynomial \( p(k) \) in the variables \( k_{1} \) and \( k_{2} \).
It is straightforward to verify that the monomials of even degree do not appear in this polynomial.
We show that, in fact, all odd-degree monomials vanish as well, except for the highest-degree term.
Our key observation is that each monomial of \( p(k) \) can be expressed as a sum of terms invariant under a dihedral group action (to introduce the action we define the evaluation function in Definition \ref{evaluation function} and develop the required identities in Lemma \ref{lem: evaluation function identities}).
These terms arise as orbit sums under irreducible representations of the dihedral group (for this observation we use Proposition \ref{prop: evaluation of product}).
If the orbit is nontrivial, the sum vanishes (this key auxiliary result is established in Lemma \ref{lem: vanishing of orbit sum}).
The coefficient of the highest degree monomial is the only fixpoint under the group action and therefore it is the only to survive.

In the following section, we introduce the necessary notation and auxiliary results required for the proof. These are then combined to establish the main result in section \ref{sec: proof of main theorem}.

It is plausible that a similar approach could yield analytical expressions for monomials of even degree.
As we develop the proof, we highlight the key differences between the odd and even cases.
Notably, the representation theory of dihedral groups differs significantly depending on whether the group order is odd or even. As a result, our argument does not extend directly to the even-degree case.
A modification of the method appears necessary to address that setting.

\subsection{Preliminaries}\label{sec: 1}
In this section, we introduce the notation and establish supporting lemmas that we use in the next section to prove the main result.
Although our reasoning relies on the structure of the dihedral group, no deep understanding of group theory will be required.
Therefore, we introduce only the necessary notation.
For a more detailed introduction of dihedral groups, we refer the reader to standard literature on group theory, such as \cite{Simon:1996aa}.

For \( n \in \bbN \), denote the additive cyclic group \( \bbZ / n \bbZ \) by \( \bbZ_{n} \).
Let \( D_{2n} \) be the dihedral group of order \( 2n \), generated by the cyclic permutation \( c \) of order \( n \) and the reflection \( r \) of order \( 2 \), along with the relation \( cr = rc^{-1} \).
The subgroup of \( D_{2n} \) generated by \( c \) is an Abelian normal subgroup, which we refer to as \( C_{n} \);
it admits a natural action on \( \bbZ_{n} \) via \( c^{k}(i) = i + k \mod{n} \) for \( i, k \in \bbZ_{n} \).

Let \( \mcA = C(\bbZ_{n}; \bbZ_{2}^{\times}) \) denote the set of functions from \( \bbZ_{n} \) to the multiplicative group \( \bbZ_{2}^{\times} = \{-1, 1\} \).
Whenever we evaluate a function \( f \in \mcA \) at a point \( k \) that is not in \( \bbZ_{n} \), we implicitly assume the periodic extension of \( f \) to \( \bbZ \), defined by \( \tilde{f}(k) = f(k \mod n) \).
For \( f \in \mcA \), set \( f^{\star} \) to be the function \( f^{\star}(i) = f(1 - i) \) for \( i \in \bbZ_{n} \),
and define an action \( \tau \) of \( D_{2n} \) on \( \mcA \) by
\[
  \tau(c^{k})f = f \circ c^{-k}, \quad \text{and}
  \quad
  \tau(r)f = f^{\star},
  \qquad \text{for all }
  f \in \mcA.
\]

Let \( \mcA(k) \subset \mcA \) denote the subset of functions for which the preimage \( f^{-1}(-1) \) has \( k \) elements; for \( k > n \), define \( \mcA(k) \) to be empty.
For each \( k \in \bbN \), the set \( \mcA(k) \) is invariant under \( \tau \) making it a \( G \)-space for \( D_{2n} \).
Define \( \mcA_{0} = \cup_{k \in 2\bbN} \mcA(k) \).
It is straightforward to see that a function \( f \) belongs to \( \mcA_{0} \) if and only if it satisfies \( \prod_{i \in \bbZ_{n}} f(i) = 1 \), which we can use as an alternative definition for \( \mcA_{0} \).

\begin{definition}
  \label{evaluation function}
  Define \( \Gamma \colon \mcA \to \bbZ_{n} \) by
  \begin{align}
    \Gamma(f) = \sum_{i \in \bbZ_{n}} i\prod_{k = 1}^{i} f(k) \mod{n}.
  \end{align}
  We call \( \Gamma \) the evaluation function.
\end{definition}

\begin{lemma}
  \label{lem: evaluation function identities}
  For \( n \in \bbN \) and \( f \in \mcA_{0} \), the evaluation function satisfies \( \Gamma(f^{\star}) \equiv -\Gamma(f) \).
  Moreover, if \( n \) is odd, there exists a non-zero constant \( F \), depending on \( f \), such that for \( k \in \bbZ_{n} \),
  \begin{align}
    \Gamma\big(\tau(c^{k})f\big)         & \equiv \big(\Gamma(f) + F \cdot k\big)\prod_{j=1}^{n-k} f(j) \pmod{n}, \\
    \Gamma\big(\tau(c^{k})f^{\star}\big) & \equiv -\Gamma(\tau(c^{-k})f) \pmod{n}.
  \end{align}
\end{lemma}

\begin{proof}

  In the following we use the convention \( \prod_{j=1}^{0}f(j) = \prod_{j=1}^{n}f(j) \).

  Fix \(f \in \mcA_{0} \).
  Note that for every subset \( S \subset \bbZ_{n} \) we have \( \prod_{i \in S}f(i) = \prod_{i \in \bbZ_{n}\setminus S} f(i) \), since \( \prod_{i\in\bbZ_{n}}f(i) =1 \).
  Using this equality in the third step of the following chain of equations we get,
  \begin{align}
    \Gamma(f^{\star})
    \equiv \sum_{i \in \bbZ_{n}} i \prod_{j=1}^{i} f(1-j)
    \equiv \sum_{i \in \bbZ_{n}} i \prod_{j=1-i}^{0} f(j)
    \equiv \sum_{i \in \bbZ_{n}} i\prod_{j=1}^{n-i} f(j)
    \equiv -\Gamma(f)
    \pmod{n}.
  \end{align}

  For the second assertion, we readily verify that
  \begin{align}
    \Gamma(\tau(c^{k})f)
     & \equiv \sum_{i \in \bbZ_{n}} (i+k) \left[ \prod_{j=1}^{i} f(j) \right]\left[ \prod_{j=1-k}^{0} f(j) \right]
     & \equiv (\Gamma(f) + F \cdot k)\prod_{j=1}^{n-k}f(j)
    \pmod{n},
  \end{align}
  with \( F = \sum_{i \in \bbZ_{n}}\prod_{j=1}^{i}f(j)\).
  If \( n \) is odd then \( F \) is a summation of odd number of terms, each of which is either \( -1 \) or \( 1 \).
  Hence, in this case \( F\neq 0 \).

  For the last assertion observe that \( \tau(c^{k})f^{\star} = \tau(c^{k}r)f=\tau(rc^{-k})f = (c^{-k}f)^{\star} \).
  Thus, by the first part, we get \( \Gamma\big(\tau(c^{k})f^{\star}\big)= \Gamma\big((\tau(c^{-k})f)^{\star}\big) = -\Gamma(\tau(c^{-k})f) \).
\end{proof}

\begin{remark}
  \label{rem: F for even n}
  If \( n \) is even, the constant \( F \) in Lemma \ref{lem: evaluation function identities} can be zero.
  For example, consider the case where \( n = 4 \) and \( f(1) = f(3) = -1 \), while \( f(0) = f(2) = 1 \)
  then \( F = \sum_{i=1}^{4} \prod_{j=1}^{i} f(j) = 0. \)
\end{remark}

Let \( [\cdot] \colon \bbZ_{2}^{\times} \to \bbZ_{2} \) be the unique group isomorphism, \( [1] = 0 \) and \( [-1] = 1 \).
\begin{proposition}
  \label{prop: evaluation of product}
  For \( n \in \bbN \) and \( f \in \mcA \) define \( a = \sum_{i\in\bbZ_{n}}[f(i)] \mod{2} \).
  Then we have
  \begin{align}
    \label{eq: evaluation of group product}
    \prod_{i =1}^{n} (r^{[f(i)]}c^{i}) = r^{a}c^{\Gamma(f)}.
  \end{align}
  If \( f \) is in \( \mcA_{0} \) then \( a = 0 \); otherwise \( a = 1 \).
\end{proposition}
\begin{proof}
  From the relation between the generators of \( D_{2n} \) it follows that for \( a,b \in \bbZ_{n} \),
  \begin{align}
    c^{a}r^{[f(b)]} = r^{[f(b)]}c^{af(b)}.
  \end{align}
  By iteratively applying this identity to the left-hand side of \eqref{eq: evaluation of group product}, we can arrange all reflections to be on the left of the cyclic permutations, yielding
  \begin{align}
    \prod_{i=1}^{n}(r^{[f(i)]}c^{i}) = r^{a} c^{\Gamma(f)},
    \qquad
    \text{with }
    a = \sum_{i\in \bbZ_{n}}[f(i)].
  \end{align}
  The constant \( a \) counts the number of elements in the preimage of \( f^{-1}(-1) \).
  For \( f \in \mcA_{0} \), this preimage has an even cardinality, implying \( a \equiv 0 \pmod{2} \); for \( f \notin \mcA_{0} \), the preimage has an odd cardinality, and \( a \equiv 1 \pmod{2} \).
\end{proof}

\begin{lemma}
  \label{lem: vanishing of orbit sum}
  Let \( n \in \bbN \) be odd, and let \( \varphi \colon D_{2n} \to B(H) \) be a faithful, irreducible representation of \( D_{2n} \) on a Hilbert space \( H \).
  For a non-constant function \( f \in \mcA_{0} \), let \( \mcO(f) \) denote the orbit of \( f \) under the action of \( D_{2n} \).
  Then we have
  \begin{align}
    \label{eq: vanishing orbit}
    \sum_{g \in \mcO(f)} \varphi\left( \prod_{i=1}^{n} r^{[g(i)]} c^{i} \right) = 0.
  \end{align}
\end{lemma}
\begin{proof}
  For \( n \) odd and \( f \in \mcA_{0} \), assume that every function in \( \mcO(f) \) is invariant under reflection.
  Then, \( \tau(r) f = f \) and \( \tau(rc^{k})f = \tau(c^{k})f \) for ever \( k \in \bbZ_{n} \).
  In particular, we have \( f = \tau(c^{-1}rc) f = \tau(c^{-2}r)f = \tau(c^{-2})f \).
  Since \( n \) is odd, \( c^{-2} \) is a generator of the cyclic group \( C_{n} \)
  and \( f \) is a fixed point of that group under the action \( \tau \).
  Consequently, \( f(i) = (\tau(c^{-i})f) (0) = f(0) \) for every \( i \in \bbZ_{n} \), which implies that \( f \) is constant.

  By assumption, \( f \) is not constant, hence \( \mcO(f) \) contains a function \( u \) that is not invariant under reflection.
  Let \( I \in \bbN \) denote the cardinality of the isotropy group \( \{ g \in D_{2n} \colon \tau(g) u = u \} \).
  Then, by Proposition \ref{prop: evaluation of product}, we have
  \begin{align}
    \label{eq: rewrite the orbit sum}
    \sum_{g \in \mcO(f)} \varphi\left( \prod_{i=1}^{n} r^{[g(i)]} c^{i} \right)
    = \sum_{g \in \mcO(f)} \varphi\left( c^{\Gamma(g)} \right)
    = I^{-1} \sum_{h \in D_{2n}} \varphi\left( c^{\Gamma(\tau(h) u)} \right).
  \end{align}
  Thus, it suffices to show that the sum on the right-hand side of \eqref{eq: rewrite the orbit sum} vanishes.

  Set \( P = \sum_{h\in D_{2n}} \varphi\left(c^{\Gamma(\tau(h)u)}\right) \in B(H) \).
  Since \( D_{2n} = \{c^{i}r^{k} \colon i \in \bbZ_{n}, \ k \in \bbZ_{2} \}\) and \( \tau(r)u\neq u \) we can split the sum in the definition of \( P \) into the summation over rotations \( \{c^{i}\} \) and reflections \( \{c^{i}r\} \).
  By Lemma \ref{lem: evaluation function identities} we then have,
  \begin{align*}
    P
    = \sum_{i \in \bbZ_{n}} \varphi\left(c^{\Gamma(\tau(c^{i})u)}\right)
    + \sum_{i \in \bbZ_{n}}\varphi\left(c^{\Gamma(\tau(c^{i})u^{\star})}\right)
    = \sum_{i\in\bbZ_{n}} \varphi(c^{F \cdot i}) \left(\varphi\big(c^{\Gamma(u)}\big) + \varphi\big(c^{-\Gamma(u)}\big)\right),
  \end{align*}
  where \( F \) is the non-zero constant \( \sum_{i\in\bbZ_{n}}\prod_{j=1}^{i}f(j) \) from Lemma \ref{lem: evaluation function identities}.

  Now, \( P \) commutes with \( \varphi(r) \) since
  \begin{align}
    \varphi(r)P
    = \sum_{i \in \bbZ_{n}} \varphi(c^{-F\cdot i})\left(\varphi\big(c^{-\Gamma(u)}\big) + \varphi\big(c^{\Gamma(u)}\big)\right) \varphi(r)
    = P \varphi(r);
  \end{align}
  and moreover, \( P \) trivially commutes with \( \varphi(c) \).
  Consequently, \( \varphi(g)P = P \varphi(g)\) for every \( g \in D_{2n} \).
  Since \( \varphi \) is an irreducible representation, Schur's lemma implies that \( P = \lambda \Id \) for some constant \( \lambda \in \bbC \).
  Then, \( \lambda \varphi(c^{F}) =\varphi(c^{F}) \lambda \Id = \varphi(c^{F})P = P = \lambda \Id\).
  Since \( F \) is non-zero and \( \varphi \) is a faithful representation it follows that \( \lambda = 0 \).
\end{proof}

\subsection{Proof of the main result}
\label{sec: proof of main theorem}
In this section we combine the results developed above to prove the main statement.

\begin{proof}[Proof of theorem \ref{thm: main}]
  For simplicity of notation let \( k_{1} = x \) and \( k_{2} = y \).
  Clearly, all entries of the final product are polynomials in \( x \) and \( y \), each of degree at most \( n \).

  Let \( \{e_{i} \colon i = 1, 2\} \) denote the canonical basis of \( \bbC^{2} \), and let \( \big\langle \cdot, \cdot \big\rangle \) denote the scalar product on \( \bbC^{2} \).
  Define the function \( g \colon \bbC^{2} \to \bbC \) pointwise by
  \begin{align}
    g(x, y) = \big\langle e_{1}, \left( \prod_{i=1}^{n} T(x,y) S(\omega^{i}) \right) e_{1} \big\rangle.
  \end{align}
  Since \( g \) is a polynomial of degree at most \( n \), it suffices to show that,
  \begin{enumerate}
    \item
      \( \frac{\partial^{k}}{\partial x^{k}} g(0, y) = 0 \) for \( 0 \leq k < n \) and all \( y \in \bbC \);
      \label{it: 3}
    \item
      \( \frac{\partial^{n}}{\partial x^{n}} g(0, y) = n!
      \) for all \( y \in \bbC \).
      \label{it: 2}
  \end{enumerate}
  Property \ref{it: 3} implies that \( g\) does not contain monomials of degree \( k < n \).
  Then, property \ref{it: 2} implies that \( g \) is independent of \( y \) and is monic, i.e., \( g(x, y) = x^{n} \).

  Let \( r, c \) be the generators of the dihedral group \( D_{2n} \) (see Section \ref{sec: 1} for notation).
  Since \( \omega = e^{i \frac{2\pi}{n}} \) is a primitive root of unity, the mapping \( \varphi \colon c^{i} r^{k} \mapsto S(\omega^{i}) X^{k} \) for \( i \in \bbZ_{n} \) and \( k \in \bbZ_{2} \) defines a faithful, irreducible representation of \( D_{2n} \) on \( \bbC^{2} \) (see, e.g., \cite[p.81]{Simon:1996aa}).

  For a fixed \( k \in \bbN \), let \( \mcA(k) \) be the set of functions in \( \mcA \) (see Section \ref{sec: 1} for notation) such that the preimage \( f^{-1}(-1) \) has \( k \) elements;
  and for \( k > n \), \( \mcA(k) \) is empty.
  Since \( \frac{\partial}{\partial x} T(0, y) = \Id \) and \( T(0, y) = y X \), it follows that the terms that appear after applying \( k \) derivatives to the product \( \prod_{i=1}^{n} T S(\omega^{i}) \) and evaluating at \( x = 0 \) are in one-to-one correspondence with functions in \( \mcA(n-k) \), via the identification \( f \leftrightarrow \prod_{i=1}^{n} X^{[f(i)]} S(\omega^{i}) \).

  In particular, by Lemma \ref{lem: evaluation function identities}, we have
  \begin{align}
    \frac{1}{k!
    }\frac{\partial^{k}}{\partial x^{k}} \left( \prod_{i=1}^{n}
    T(x,y) S(\omega^{i}) \right)
     & = y^{n-k} \sum_{f \in \mcA(n-k)} \prod_{i=1}^{n} X^{[f(i)]} S(\omega^{i}) \nonumber                               \\
     & = y^{n-k} \sum_{f \in \mcA(n-k)} \varphi\left( \prod_{i=1}^{n} r^{[f(i)]} c^{i} \right) \label{eq: k-derivatives} \\
     & = y^{n-k} \sum_{f \in \mcA(n-k)} \varphi\left( r^{n-k} c^{\Gamma(f)} \right)
    = y^{n-k} X^{n-k} \sum_{f \in \mcA(n-k)} S(\omega^{\Gamma(f)}).
    \nonumber
  \end{align}
  Since \( S(z) \) is diagonal in the \( \{e_{1}, e_{2}\} \)-basis, and \( X \) is a self-adjoint unitary, it follows that whenever \( n - k \) is odd i.e. \( \mcA(n-k) \not\subset \mcA_{0} \), \( \frac{\partial^{k}}{\partial x^{k}} g(0, y) \) vanishes identically.

  If \( k < n \) is even, then \( n - k \) is odd and the above consideration implies that \( \frac{\partial^{k}}{\partial x^{k}} g(0, y) = 0 \) for every \( y \in \bbR \).
  If \( k < n\) is odd, then \( \mcA(n-k) \subset \mcA_{0} \) is a \( G \)-space for the dihedral group under the action \( \tau \) and is therefore a disjoint union of orbits \cite[Prop.
    I.2.2]{Simon:1996aa}.
  Moreover, \( \mcA(n-k) \) does not contain constant functions, and thus by Lemma \ref{lem: vanishing of orbit sum}, the sum over each orbit vanishes individually.
  This shows that \eqref{eq: k-derivatives} vanishes (without even taking the scalar product) when \( k < n \) is odd.
  In any case, we get \( \frac{\partial^{k}}{\partial x^{k}} g(0, y) = 0\) for \( 0 \leq k < n \) which establishes \ref{it: 3}.

  For \ref{it: 2}, observe that \( \frac{n(n-1)}{2} \) is an integer multiple of \( n \).
  For \( y \in \bbC \), we thus have
  \begin{align}
    \frac{\partial^{n}}{\partial x^{n}}g(0,y)
    = n!
    \big\langle e_{1}, \prod_{i=1}^{n}
    S(\omega^{i}) e_{1} \big\rangle
    = n!
    \big\langle e_{1}, S(\omega^{\frac{n(n-1)}{2}}) e_{1} \big\rangle
    = n! \big\langle e_{1}, e_{1} \big\rangle
    = n!.
  \end{align}
  This concludes the proof.
\end{proof}

\begin{remark}
  We stated the result for the case when \( n \) is odd.
  Naturally, one might ask whether the similar reasoning can be extended to the case when \( n \) is even.
  Unfortunately, this is not the case.

  Lemma \ref{lem: vanishing of orbit sum} explicitly requires the function \( f \) to be non-constant.
  When \( n \) is odd, the only constant function in \( \mcA_{0} \) is \( f = 1 \).
  This case appears in the above theorem when establishing property \ref{it: 2}, resulting in a non-vanishing coefficient for the monomial \( x^{n} \).
  However, when \( n \) is even, the function \( f = -1 \) also belongs to \( \mcA_{0} \).
  As a result, \( \varphi\left( \prod_{i=1}^{n} r^{[-1]} c^{i} \right) \) does not vanish, thereby contradicting property \ref{it: 3} and, in particular, implying \( g(0, y) \neq 0 \).

  Moreover, as noted in Remark \ref{rem: F for even n}, for an even \( n \) the constant \( F \) can vanish, thereby rendering Lemma \ref{lem: vanishing of orbit sum} inapplicable.
  Thus, in this case even if \( f \) is not constant the sum over the orbit \( \mcO(f) \) need not be zero.

\end{remark}

\section*{Acknowledgements}
We thank Satoshi Ejima for constructive discussions and the original motivation to address the problem.
This project was made possible by the DLR Quantum Computing Initiative and the Federal Ministry for Economic Affairs and Climate Action.

\printbibliography
\end{document}